\documentclass{LMCS}
\usepackage{hyperref,enumerate}
\usepackage{xspace}
\usepackage{amsmath}
\usepackage{amssymb}
\usepackage{bm}

\DeclareFontFamily{U}{bbold}{}
\DeclareFontShape{U}{bbold}{m}{n}
   {  <5> <6> <7> <8> <9> gen * bbold
      <10> <10.95> bbold10
      <12> <14.4> bbold12
      <17.28> <20.74> <24.88> bbold17
   }{}
\DeclareSymbolFont{bbold}{U}{bbold}{m}{n}
\DeclareMathSymbol{\BssD}{3}{bbold}{"01}
\DeclareMathSymbol{\BssS}{3}{bbold}{"06}
\DeclareMathSymbol{\BssP}{3}{bbold}{"05}
\DeclareMathSymbol{\BssReduceq}{3}{bbold}{"3C}

\newcommand{\reduceq}{\preceq}

\newcommand{\IR}{\mathbb{R}}
\newcommand{\IA}{\mathbb{A}}
\newcommand{\IF}{\mathbb{F}}
\newcommand{\IE}{\mathbb{E}}
\newcommand{\IS}{\mathbb{S}}

\newcommand{\IW}{\mathbb{W}}
\newcommand{\IB}{\mathbb{B}}

\newcommand{\IQ}{\mathbb{Q}}
\newcommand{\IL}{\mathbb{L}}

\newcommand{\IK}{\mathbb{K}}
\newcommand{\IZ}{\mathbb{Z}}
\newcommand{\IX}{\mathbb{X}}
\newcommand{\IH}{\mathbb{H}}
\newcommand{\IN}{\mathbb{N}}

\newcommand{\Universe}{\mathcal{U}}
\newcommand{\BSS}{{\rm BSS}\xspace}
\newcommand{\BCSS}{{\rm BSS}\xspace}
\newcommand{\ball}{B}
\newcommand{\dist}{\operatorname{dist}}
\newcommand{\dom}{\operatorname{dom}}
\newcommand{\range}{\operatorname{range}}

\newcommand{\trdeg}{\operatorname{trdeg}}
\newcommand{\cf}[1]{\mathbf{1}_{#1}}
\newcommand{\calC}{\mathcal{C}}

\newcommand{\calF}{\mathcal{F}}
\newcommand{\calM}{\mathcal{M}}
\newcommand{\calN}{\mathcal{N}}
\newcommand{\calR}{\mathcal{R}}

\newcommand{\calT}{\mathcal{T}}

\newcommand{\sign}{\operatorname{sign}}
\newcommand{\sepigraph}{\operatorname{s-epigraph}}
\newcommand{\shypograph}{\operatorname{s-hypograph}}
\newcommand{\person}[1]{\textsc{#1}}
\newcommand{\aname}[1]{\emph{#1}}
\newcommand{\mycite}[2]{{\rm\cite[\textsc{#1}]{#2}}}
\newcommand{\closure}[1]{\overline{#1}}
\newcommand{\Cantor}{\calC}

\newcommand{\KleeneS}{\Sigma}
\newcommand{\KleeneP}{\Pi}
\newcommand{\KleeneD}{\Delta}
\newcommand{\BorelS}{\bm{\Sigma}}
\newcommand{\BorelP}{\bm{\Pi}}
\newcommand{\BorelD}{\bm{\Delta}}
\newcommand{\Card}{\operatorname{Card}}
\newcommand{\Fsigma}{\text{F}_{\sigma}}
\newcommand{\Gdelta}{\text{G}_{\delta}}
\newcommand{\COMMENTED}[1]{}

\newtheorem{fac}[thm]{Fact}

\theoremstyle{plain}\newtheorem{question}[thm]{Question}

\def\doi{7 (3:11) 2011}
\lmcsheading%
{\doi}
{1--20}
{}
{}
{Nov.~29, 2010}
{Sep.~\phantom02, 2011}
{}   

\begin{document}

\title[Real Analytic Machines and Degrees]{%
Real Analytic Machines and Degrees: \\
A Topological View on Algebraic Limiting Computation\rsuper*}
\author[T.~G\"{a}rtner]{Tobias G\"{a}rtner\rsuper a}
\address{{\lsuper a}Universit\"{a}t des Saarlandes, GERMANY}
\email{tobi.gaertner@web.de}
\author[M.~Ziegler]{Martin Ziegler\rsuper b}
\address{{\lsuper b}Technische Universit\"{a}t Darmstadt, GERMANY}
\email{ziegler@mathematik.tu-darmstadt.de}
\thanks{Supported by the 
\emph{Deutsche Forschungsgemeinschaft}
with project \texttt{Zi\,1009/2-1}. We are grateful to \person{Arno Pauly}
for contributing Theorem~\ref{t:Arno}; and to \person{Vassilios Gregoriades}
for raising questions that led to Section~\ref{s:Composition}.
Further thank is due to the anonymous referees 
who also helped simplify the proof to Example~\ref{x:Pairing}
and provided several further constructive remarks.}
\titlecomment{{\lsuper*}A preliminary version of this work had appeared
in \emph{Proc. CCA 2011}, EPTCS vol.\textbf{24},
\url{arXiv:1006.0398v1}}


\begin{abstract}
We study and compare in degree-theoretic ways
(iterated Halting oracles analogous to Kleene's
arithmetical hierarchy, and the Borel hierarchy of
descriptive set theory) the capabilities and limitations
of three models of real computation:
\BSS machines (aka real-RAM) 
and strongly/weakly analytic machines as
introduced by Hotz et al. (1995).
\end{abstract}
\maketitle
\section{Introduction} \label{s:Intro}
The Turing machine as standard model of (finite) computation and 
computational complexity over discrete universes
like $\IN$ or $\{0,1\}^*$ 
suggests two both natural but distinct 
ways of extension to real numbers:
Already Alan Turing \cite{Turing2} considered \emph{in}finitely
long calculations producing as output a sequence of integer fractions
$r_n/s_n\in\IQ$ (i.e. discrete objects) approximating some real $x\in\IR$
up to error $2^{-n}$.

The first model (dominant in Recursive Analysis)
reflects that actual digital computers can operate in each
step only on finite information, and in particular with limited precision 
on real numbers \mycite{Theorem~4.3.1}{Weihrauch}.
It is, however, often criticized \cite{Koepf} for the consequence 
that any computable function must necessarily be continuous. 
In the second model, simple discontinuous functions like 
Heaviside or Gau\ss{} Staircase are computable, whereas 
intricate and intuitively non-computable functions 
(such as the characteristic of Mandelbrot's fractal set) 
can indeed be proven uncomputable.
Criticism of this model arises from its ability to compute
certain pathological functions \mycite{Example~9.7.2}{Weihrauch}
but not as simple functions as square root or exponential.
A formal introduction to this model is deferred to Section~\ref{s:BSS}.
We refer to \cite{Ning,Boldi,Emperor} for comparisons between both models.

\subsection{Oracle Computation and Turing Degrees} \label{s:Post}
Classical (i.e. discrete) computability and complexity theory
had, almost from the very beginning \cite{Turing3}, started
considering machines with access to oracles: not so much to
model actual computational practice, but because it permits to formally compare
problems according to their degree of uncomputability or complexity
\cite{Soare,Papadimitriou}. 
For a universe $\Universe$ like $\IN$ or,
equivalently (cmp. Example~\ref{x:Pairing} below), $\{0,1\}^*$,
the Halting problem $H\subseteq\Universe$ 
(sometimes denoted as \emph{jump} of $\emptyset$) for instance
becomes trivially decidable when granted oracle access to $H$;
whereas the iterated Halting problem (or \emph{jump} of $H$)
\begin{equation} \label{e:Halting2}
H^H \;:=\; \big\{ \langle M,\vec x\rangle:
 \text{oracle Turing machine $M^H$ 
  terminates on input $\vec x\in\Universe$}\big\} \enspace , 
\end{equation}
that is the question of termination of a given such Turing machine
with $H$-oracle, remains undecidable by an $H$-oracle Turing machine.

The class of semidecidable problems is often denoted as
$\KleeneS_1$; $\KleeneP_1$ are their complements, and
$\KleeneD_1=\KleeneS_1\cap\KleeneP_1$ the decidable problems.
$\KleeneS_2$ are problems of the form
\begin{equation} \label{e:Kleene2}
\;\big\{ \vec x\in\Universe \;\big|\; \exists\vec y\in\Universe
\;\forall\vec z\in\Universe: 
  \;\langle\vec x,\vec y,\vec z\rangle\in T \big\}
\end{equation}
with decidable $T\subseteq\Universe$; $\KleeneP_2$ consists of complements
of $\KleeneS_2$-problems, that is problems defined by
``$\forall\exists$''-formulas. Similarly,
``$\exists\forall\exists$'' defines $\KleeneS_3$;
and so on: \person{Kleene}'s \aname{Arithmetical Hierarchy}.
\person{Post}'s \aname{Theorem} \mycite{Section~IV.2}{Soare} 
asserts for a set $S\subseteq\Universe$ 
the following to be equivalent:
\begin{enumerate}[a)]
\item 
 $S$ is semidecidable relative to the Turing Halting Problem $H$.
\item
 $S\in\KleeneS_2$, i.e. has the form (\ref{e:Kleene2})
 with decidable $T$.
\item
 $S$ is many-one reducible to the iterated Halting problem 
 (written ``$S\reduceq H^H$'').
\end{enumerate}\vspace{3 pt}

\noindent Analogously, $S$ is semidecidable relative to $H^H$ ~iff~
$S\in\KleeneS_3$ ~iff~ $S$ is many-one reducible to $H^{H^H}$.
Some natural $\KleeneS_2$-complete and $\KleeneS_3$-complete problems
are identified in \mycite{Section~IV.3}{Soare}.

Similar investigations have been 
pursued in both the \BSS model \cite{Cucker} and in Recursive
Analysis \cite{Ho,Xizhong}. However, and as opposed to the discrete
setting, oracles turn out to be of limited help in the case of real functions:
both square root and exponential still remain uncomputable to a \BSS machine; 
and computability in Recursive Analysis 
remains restricted to continuous functions \cite{ToCS,CCA06}. 

\subsection{Limiting Computation}
This subsumes calculations whose result occurs 
after $\geq\omega$ steps. It may include transfinite cases \cite{Hamkins}
although we shall restrict to computations producing a 
countable sequence of outputs that `converge' to the result.
For finite binary strings (i.e. w.r.t. the discrete topology), 
this appears in the literature as 
\emph{limiting recursion} \cite{Gold,Schmidhuber},
\emph{inductive algorithms} \cite{Burgin}, or
\emph{trial-and-error predicates} \cite{Putnam}.

On continuous universes like Cantor or Baire space or reals, 
a classical Turing machine
processing finite information in each step needs an infinite amount of time.
This led to the field of Recursive Analysis ($\IR^d$) and
\person{Weihrauch}'s Type-2 Theory ($\{0,1\}^\omega$, $\IN^\omega$).
Here, convergence is required to be \emph{effective} in the sense
that the $n$-approximate output differs from the ultimate result
by no more than an absolute error of $2^{-n}$; equivalently,
the approximations are to be accompanied by absolute error bounds
tending to 0.

On the other hand, relaxing the output 
(but not input \mycite{Section~6}{Naive})
to merely converging approximations
does render some discontinuous functions computable;
and in fact corresponds to climbing up the effective Borel
Hierarchy \cite{EffBorel}:

\subsection{Topological Complexity and Borel Degrees} \label{s:Borel}
Recall 
the Borel Hierarchy on an arbitrary topological space $X$:
$\BorelS_1(X)$ denotes the family of open subsets of $X$,
$\BorelP_1(X)$ that of closed sets (i.e. complements of open ones),
$\BorelS_2(X)$ the class of countable unions of closed sets (aka $\Fsigma$), 
$\BorelP_2(X)$ that of countable intersections of open sets (aka $\Gdelta$,
i.e. complements of $\Fsigma$), and iteratively
$\BorelS_{k}(X)$ the class of countable unions of
$\BorelP_{k-1}$-sets,
and $\BorelP_{k}(X)$ their complements;
$\BorelD_k(X)=\BorelS_k(X)\cap\BorelP_k(X)$.

We shall frequently make use of the folklore

\begin{fac} \label{f:Borel}\hfill
\begin{enumerate}[a)]
\item It holds
$\IQ\in\BorelS_2(\IR)\setminus\BorelP_2(\IR)$.
\item It holds
$\IQ\times(\IR\setminus\IQ)\not\in\BorelP_2(\IR^2)\cup\BorelS_2(\IR^2)$.
\item If $X\in\BorelS_k(Y)$ and $Y\in\BorelS_k(Z)$,
  then $X\in\BorelS_k(Z)$; similarly for $\BorelP_k$.
\end{enumerate}
\end{fac}
\proof\hfill
\begin{enumerate}[a)]
\item $\IQ$ is a countable union of (closed) singletons,
hence in $\BorelS_2$. Suppose $\IQ\in\BorelP_2$,
i.e. $\IQ^\complement=\bigcup_n A_n$ with closed $A_n$.
But $\IR\setminus\IQ$ being of second category
according to \aname{Baire's Theorem},
there has to exist at least one with nonempty interior:
$\emptyset\neq\closure{A_n}^\circ=A_n^\circ\subseteq(\IR\setminus\IQ)^\circ$,
a contradiction. 
\item
$\IQ\times(\IR\setminus\IQ)$ contains as sections both
$\IQ\not\in\BorelP_2$ and $\IR\setminus\IQ\not\in\BorelS_2$,
hence cannot be in $\BorelP_2$ nor in $\BorelS_2$.
\item
Since $Y$ is equipped with the relative topology of $Z$,
induction on $k$ shows that $X\in\BorelS_k(Y)$ implies
$X=Y\cap X'$ for some $X'\in\BorelS_k(Z)$; and the latter
class is closed under finite intersection.\qed
\end{enumerate}

\noindent If $X$ is a polish space, the Borel hierarchy is strict 
and contains complete members \cite{Kechris}.
According to \textsf{Alexandrov's Theorem}, every
$\Gdelta$--subset of a Polish space (such as $\IR^d$) 
is again Polish.
In the sequel, all spaces $X$ under consideration will be
(not necessarily closed)
subspaces of some $\IR^d$ ($d\in\IN$) equipped with
the Euclidean topology. 

Also recall that continuity of a (total) function
$f:X\to Y$ means that preimages $f^{-1}[V]$ of open 
sets $V\subseteq Y$ are open in $X$;
and, more generally, $f$ is called $\BorelS_k$-measurable
if preimages of open subsets of $Y$ are in $\BorelS_k(X)$.

As common in classical computability theory, 
partial functions $f:\subseteq\IR^d\to\IR$ 
arise naturally also in the real case.
We distinguish between the Borel complexity of their
domain $\dom(f)\subseteq\IR^d$ and that of the 
\emph{total} function $f:\dom(f)\to\IR$; 
cmp. Remark~\ref{r:partial} below.

\subsection{Blum-Shub-Smale Machines} \label{s:BSS}
The \BSS model (over $\IR$) considers real numbers as entities
that can be read, stored, output,
added, subtracted, multiplied, divided, and compared exactly.
It captures the semantics of, e.g., the \texttt{FORTRAN} programming
language and essentially coincides with the \textsf{real-RAM} model
underlying, e.g., Algorithmic Geometry \cite{compGeom}.
Specifically, a program consists of a sequence of 
arithmetic instructions (\texttt{+,-,$\times$,$\div$})
and branchings based on tests (\texttt{=,<}).
A countably infinite sequence of working registers 
can be accessed directly or via indirect addressing 
through dedicated integer-valued index registers. 
Assignments ($a$\texttt{:=}$b$) may copy data between registers
or initialize them to one of finitely many real constants listed
in the program. 
Upon start, the input vector $(x_1,\ldots,x_d)\in\IR^d$ is provided
in the real registers and, for uniformity purposes, its dimension $d$
in the index registers. If the machine terminates within finitely
many steps, the contents of the real registers (up to index given
by the index register) is considered as output.
For further details we refer to \cite{BCSS} or
\mycite{Definition~1.1}{Cucker}.

\begin{defi} \label{d:Partial}
For $\IX\subseteq\IR^*$, a set $\IL\subseteq\IX$ is called
\emph{\BSS-decidable \textsf{in} $\IX$} if some \BSS machine can,
given any $\vec x\in\IX$, report within finite time which 
one of $\vec x\in\IL$ or $\vec x\not\in\IL$ holds.
$\IL$ is \emph{\BSS-semidecidable \textsf{in} $\IX$} if some
\BSS machine terminates for every $\vec x\in\IL$ and
does not terminate 
(`\emph{diverges}') for every $\vec x\in\IX\setminus\IL$.
A (possibly partial) function $f:\subseteq\IR^*\to\IR^*$ is
\emph{\BSS-computable \textsf{in} $\IX$} if some \BSS
machine, on inputs $\vec x\in\dom(f)\cap\IX$, terminates
with output $f(\vec x)$ and
diverges on inputs $\vec x\in\IX\setminus\dom(f)$.
\end{defi}
Note that this common notion \mycite{Definition~4.2}{BCSS}
ignores the behavior on inputs outside of $\IX$.
In the case $\IX=\IR^*$, 
we simply speak of \BSS (semi-)decidability/computa\-bi\-lity of $\IL/f$.

\subsection{The Arithmetical Hierarchy for \BSS Machines}
\mycite{Theorems~2.11+2.13}{Cucker} has succeeded in generalizing 
Post's Theorem to oracle \BSS machines,
however based on entirely different arguments:

\begin{fac} \label{f:Cucker}
For a set $\IS\subseteq\IR^*$ the following are equivalent:
\begin{enumerate}[a)]
\item
 $\IS$ is semidecidable by a \BSS machine with oracle access to $\IH$.
\item
 There exists a \BSS-decidable set $\IW\subseteq\IR^*$ such that
\begin{equation} \label{e:Sigma2}
\IS\;=\;\big\{ \vec x\in\IR^* \;\big|\; \exists y\in\IN
\;\forall z\in\IN: \; (\vec x,y,z)\in\IW\big\} 
\enspace . \end{equation}
\item
 $\IS$ is \BSS many-one reducible to the \emph{iterated} \BSS Halting problem:
\[ \IS \;\BssReduceq\; \IH^\IH \;:=\; \big\{ (\langle\calM\rangle,\vec x):
 \text{oracle \BSS machine $\calM^{\IH}$ 
  terminates on input $\vec x$}\big\} \enspace . \]
\end{enumerate}
\end{fac}\vspace{3 pt}
\noindent Observe that, in spite of referring to a real complexity class,
quantifiers in Equation~(\ref{e:Sigma2}) range over integers.

\subsection{Analytic Machines} \label{s:Analytic}
In \cite{Vierke,Chadzelek}, a third model and kind of synthesis of the above two 
had been proposed: An \emph{analytic machine} is essentially a 
\BSS machine (i.e. with \emph{exact} arithmetic operations and tests)
permitted to \emph{approximate} the output; either with 
(\emph{strongly analytic}) 
or without (\emph{weakly analytic}) error bounds.
More precisely, let $\|\vec y\|:=\sum_i |y_i|$ denote the 1-norm
and $\dim(\vec y)=d$ for $\vec y\in\IR^d$.
Then computing $f:\IR^*\to\IR^*$ means
producing, on input $\vec x\in\dom(f)$, some infinite
sequence $\vec y_n\in\IR^*$ with, for $\vec y:=f(\vec x)$,

\begin{equation} \label{e:Analytic}
\begin{gathered}
\dim(\vec y_n)=\dim(\vec y) \qquad\text{ and }\qquad
\|\vec y_n-\vec y\|\leq2^{-n}
\quad\text{for all $n\in\IN$ \qquad (strong)} \\
\dim(\vec y_n)=\dim(\vec y) \quad\text{ for all but
  finitely many $n$ and}
\;\;\lim\nolimits_n\vec y_n=\vec y \qquad\text{ (weak)}
\end{gathered}
\end{equation}
Another variation depends on whether the `program' may employ
finitely many pre-stored real constants or not.
(Several further variants considered in \cite{Chadzelek} are
outside the scope of the present work\ldots)
Now Heaviside and square root and exponential function are easily 
seen to be computable by a strongly analytic machine. In fact, every
computation in Recursive Analysis or by a \BSS machine 
(without constants)
can be simulated on a strongly analytic machine (without constants).

\begin{rem} \label{r:partial}
Recall that a \BSS machine computing a partial function
$f:\subseteq\IR^*\to\IR^*$ is required to diverge 
on inputs $\vec x\not\in\dom(f)$. This convention is
common in the \BSS community and corresponds to what
in Recursive Analysis is called \emph{strong computability}
\mycite{Exercise~4.3.18}{Weihrauch}. Similarly,
a strongly/weakly analytic machine computing $f$ 
must violate Equation~(\ref{e:Analytic})
for $\vec x\not\in\dom(f)$. 
\end{rem}
Note that in the strong case, this amounts to the
output of either only a finite sequence 
(i.e. a terminating computation or an indefinite one which, 
however, eventually ceases to print further approximations)
or of one for which the following fails:
\begin{equation} \label{e:Variation}
\|\vec y_n-\vec y_m\| \;\leq\; 2^{-n}+2^{-m}  
\quad\forall n,m
\enspace .
\end{equation}
As a consequence, the domain (not to mention its complement)
of a strongly analytically computable partial function
in general need not be \BSS-semidecidable
(i.e. the domain of a \BSS-computable function):

\begin{exas} \label{x:Domain}\hfill
\begin{enumerate}[a)]
\item
There exists a \BSS-computable
function $f:\subseteq\IR^*\to\{1\}$ with $\dom(f)=\IH$,
the Halting problem for $\BSS$ machines:
\[ \IH \quad:=\quad
\big\{ (\langle\calM\rangle,\vec x): 
\text{constant-free \BSS machine $\calM$ terminates
on input $\vec x\in\IR^*$} \big \} \]
\item
There is a function $g:\subseteq\IR^*\to\{1\}$
computable by a strongly analytic machine
with $\dom(g)=\IH^\complement$.
\end{enumerate}
\end{exas}
\proof\hfill
\begin{enumerate}[a)]
\item
Define $f(\langle\calM\rangle,\vec x):=1$ for
$(\langle\calM\rangle,\vec x)\in\IH$,
$f(\langle\calM\rangle,\vec x):=\bot$ otherwise.
A universal \BSS machine can obviously compute this.
\item
Define $g(\langle\calM\rangle,\vec x):=\bot$ for
$(\langle\calM\rangle,\vec x)\in\IH$,
$g(\langle\calM\rangle,\vec x):=1$ otherwise.
Consider a variant of the universal \BSS machine which,
upon input $(\langle\calM\rangle,\vec x)$, simulates 
the computation of $\calM$ on $\vec x$ and at each step
outputs 1s but switches to alternating its outputs between 0s and 1s at each step
when $\calM$ terminates.
\qed\end{enumerate}
%
We will explore this discrepancy more closely in
Corollary~\ref{c:StrongBorel}a+d).
Technically this means that certain of our results (have to)
refer to \emph{some extension} of a given partial function;
cmp. Definition~\ref{d:Partial}.

\subsection{Overview}
Section~\ref{s:Strongly} explores the power of strongly analytic machines
by comparison to classical \BSS machines.
Roughly speaking, it turns out that \BSS-computable
total functions are $\BorelD_2$-measurable and
partial functions have $\BorelS_2$-measurable domains;
whereas total functions computable by strongly analytic
machines cover all $\BorelS_1$-measurable (i.e. continuous)
ones while including also some $\BorelS_2$-measurable ones---but
in order to cover all of them, the composition of two such functions
is sufficient and in general necessary.
Partial functions here have $\BorelP_3$-measurable domains.

We then (Section~\ref{s:Weakly}) proceed to weakly analytic machines. 
These are characterized as strongly analytic ones 
relative to the \BSS halting problem by establishing
a real variant of the Shoenfield Limit Lemma.
The proof is assisted by a particular notion of
weak semidecidability (Section~\ref{s:WeakSemi})
many-one equivalent to the iterated (i.e. jump of the)
\BSS Halting problem.
Similar to the strongly analytic
case from Example~\ref{x:Domain},
Section~\ref{s:Convergence} reveals the convergence
problem (i.e. domain of a function computable by a weakly analytic
machine) to not be weakly semidecidable in general.
In fact we show the complement of this problem equivalent 
to the jump of the iterated \BSS Halting problem.

\section{Exploring the Power of Strongly Analytic Machines} \label{s:Strongly}
Both Recursive Analysis and the \BSS model 
exhibit many properties similar to the classical theory of 
computation---although not all of them. 
Analytic Machines behave even more nicely in this sense:

\begin{exas} \label{x:Pairing}
A major part of discrete recursion theory relies on the existence
of a bicomputable pairing function, that is a bijective encoding 
\[
\IN\times\IN \;\ni\; (x,y) \;\mapsto\; \langle x,y\rangle \;\in\; \IN \]
and decoding of pairs of integers into a single one,
both computable by a Turing machine. 
($\langle x,y\rangle:=2^{x-1}\cdot(2y-1)$ for instance will do\ldots)
\\
This significantly differs from the real
case where 
there is no \BCSS-computable (even local) real pairing function:
\begin{enumerate}[a)]
\item
Let $\emptyset\neq U\subseteq\IR^2$ be open and $f:U\to\IR$ injective.
Then $f$ is not \BCSS-computable.
\item
Let $\emptyset\neq U\subseteq\IR^2$ be open and
$g:\subseteq\IR\to U$ surjective.
Then $g$ is not \BCSS-computable relative to any (set) oracle.
\item
There is, however, a strongly analytic machine without constants
computing a total bijection $h:\IR^2\to\IR$ and its inverse.
\end{enumerate}
\end{exas}
\proof\hfill
\begin{enumerate}[a)]
\item
Suppose $f$ is computable by some \BCSS machine;
then its path decomposition yields a non-empty open ball $\ball\subseteq U$
such that $f|_{\ball}$ is a rational function and in particular continuous.
But a continuous $f$ from a convex open $\ball\subseteq\IR^2$ to $\IR$
cannot be injective: 
Consider $\vec u,\vec v\in\ball$ and two disjoint paths $g,h:[0,1]\to\ball$ 
connecting them, that is with
$g(0)=\vec u=h(0)$ and $g(1)=\vec v=h(1)$ and
$g([0,1])\cap h([0,1])=\{\vec u,\vec v\}$. By continuity, it
follows from the intermediate value theorem that $f \circ g$ and $f \circ h$ must
have a common value on the open interval $(0,1)$. This contradicts the assumption
of $f$ being injective.
\item
Recall that for fields $\IF\subseteq\IE$,
the transcendence degree of $S\subseteq\IE$ (over $\IF$)
$\trdeg_{\IF}(S)$ denotes the cardinality of a largest
subset of $S$ algebraically independent over $\IF$;
equivalently: of a least $T\subseteq S$ such that
$\IF(S)$ is algebraic over $\IF(T)$.
In particular,
a finite $S$ is algebraically independent 
over a finite field extension $\IF(T)$ ~iff~
$\trdeg_{\IF}(S\cup T)=\Card(S)+\trdeg_{\IF}(T)$
\mycite{Proposition~5.1.2}{Cohn3}.
\\
Now observe that the output of a \BCSS-machine $\calM$ with constants $c_1,\ldots,c_d$
on input $x_1,\ldots,x_n$ is limited to the rational field extension
$\IQ(c_1,\ldots,c_d,x_1,\ldots,x_n)$. 
On the other hand, the transcendence degree 
$\trdeg_{\IQ}(\IR)$ of $\IR$ over $\IQ$ is infinite.
In particular, there exist $y,z\in\IR$ with $\{y,z\}$ algebraically independent
over $\IQ(c_1,\ldots,c_d)$. Now $\IQ^2$ is dense in $\IR^2$ and $U\subseteq\IR^2$
is open; hence there are $p,q\in\IQ$ with $(y+p,z+q)\in U$; 
and $\{y+p,z+q\}$ remains algebraically independent over $\IQ(c_1,\ldots,c_d)$.
So suppose there is a machine $\calM$ with constants $c_1,\ldots,c_d$ computing $g$
and in particular $(y+p,z+q)=g(x)$ for some $x$. Then 
$y+p,z+q\in\IQ(c_1,\ldots,c_d,x)$ by the above observation,
hence $\trdeg_{\IQ}(y+p,z+q,c_1,\ldots,c_d)\leq 
\trdeg_{\IQ}(c_1,\ldots,c_d,x)\leq
\trdeg_{\IQ}(c_1,\ldots,c_d)+1$; 
whereas algebraic independence of 
$\{y+p,z+q\}$ over $\IQ(c_1,\ldots,c_d)$ requires 
$\trdeg_{\IQ}(y+p,z+q,c_1,\ldots,c_d)
=2+\trdeg_{\IQ}(c_1,\ldots,c_d)$: contradiction.
\item
Observe that the mapping
\[ [0,1) \;\ni\; x=\sum\nolimits_{n=1}^\infty b_n 2^{-n} 
\;\mapsto\; (b_1,b_2,\ldots,b_n,\ldots)\in \{0,1\}^\IN   \enspace, \]
extracting from a given real number its\footnote{\label{myfn}%
Dyadic rationals $x=(2\ell+1)/2^k$ have two distinct such
expansions. For the purpose of well-definition, we here
refer to the one with only finitely many 1s.} binary 
expansion is computable digit by digit by (exact) 
comparison: For $n=1$, let $b_n:=1$ in case $x\geq1/2$
and $b_n:=0$ otherwise;
then iterate with $x\mapsto 2x-1$ and $n\mapsto n+1$.
\\
Conversely, a strongly analytic machine (but no \BSS machine)
can encode any binary sequence $b_n$ (say, of intermediate results)
into a real number $\sum\nolimits_{n=1}^\infty b_n 2^{-n}\in[0,1]$
by approximations $\sum\nolimits_{n=1}^N b_n 2^{-n}$
up to error $2^{-N}$.
\\
We will thus decode both $x,y\in[0,1)$ into their
binary expansions $a_n$ and $b_n$, merge the latter
into $c_{2n}:=a_n$ and $c_{2n-1}:=b_n$, and then 
re-code $(c_m)_{_m}$ into $z\in[0,1)$.
\\
This procedure obviously extends from $[0,1)$ to
$[0,\infty)$ with binary expansion \linebreak $\sum_{n=-k}^\infty a_n 2^{-n}$
and, incorporating also $\sign(x)$ and $\sign(y)$, 
to a pairing function $h:\IR\times\IR\to\IR$.
\\
Its inverse can be computed similarly by a strongly analytic machine:
From a given $z\in\IR$, extract from its binary expansion 
the (sub)sequence of those digits with even/odd index and
compose them into (approximations up to error $2^{-n}$ of) 
$x,y\in\IR$ with $h(x,y)=z$.
\qed\end{enumerate}
It is clear that a \BSS machine without constants
cannot compute the constant function $f(x)\equiv c$
unless $c\in\IQ$; and oracles do not help.
This is different for analytic machines:

\begin{prop} \label{p:Constants}
Let $\IH$ 
denote the Halting problem for \BSS machines
from Example~\ref{x:Domain}a).
To every strongly/weakly
analytic machine $\calM$ with \emph{recursive} constants,
there exists a \linebreak
strongly/weakly
analytic oracle machine $\calN$ 
such that $\calN^\IH$ is equivalent to $\calM$.
\end{prop}
\begin{proof}
Let $c_1,\ldots,c_k\in\IR$ denote the constants of an
analytic machine $\calM$.
Observe that each computation of $\calM$ on
some input $\vec x\in\IR^d$ can be described
as an infinite sequence of elementary operations 
(arithmetic on two intermediate results, 
branch based on testing some intermediate result, 
output of an intermediate result); where each intermediate
result is a rational function
$R(x_1,\ldots,x_d,c_1,\ldots,c_k)$
with rational coefficients.
Now given such an input $\vec x$, let $\calN^\IH$
symbolically record the intermediate calculations 
$R$ performed by $\calM$ up to the first test
``$R(\vec x,\vec c):0$?''
Since $\vec c$ is presumed recursive,
so is $R(\vec x,\vec c)$ and can be output
up to any desired precision whenever
$\calM$ would output $R(\vec x,\vec c)$ exactly.
Another consequence, both ``$R(\vec x,\vec c)<0$'' and
``$R(\vec x,\vec c)>0$'' are semidecidable;
thus by querying $\IH$, $\calN^\IH$ can decide
``$R(\vec x,\vec c)=0$'' and proceed with the
simulation of $\calM$'s control flow accordingly. 
\end{proof}
The following technical tool will be used in the sequel:
\begin{fac} \label{f:Dist}
For a set $S\subseteq\IR^k$, consider
its \aname{distance function}
\begin{equation} \label{e:Distance}
  \dist(\cdot;S):\IR^k\to[0,\infty], 
\quad \vec x\mapsto\inf\big\{\|\vec x-\vec y\|:\vec y\in S\big\} \enspace .
\end{equation}
\begin{enumerate}[a)]
\item
The function $\dist(\cdot;S)$ is always continuous.
\item[b)]
If $S$ is closed, the infimum in Equation~(\ref{e:Distance})
is a minimum, i.e. attained.
\item[c)]
For closed $A,B\subseteq\IR^k$,
it holds $\dist(\vec x;A\cup B)=\min\{\dist(\vec x;A),\dist(\vec x;B)\}$.
\end{enumerate}
\end{fac}

\begin{exa} \label{x:DistCantor}
\person{Cantor}'s \aname{Excluded Middle}
set $\Cantor\subseteq[0,1]$ is closed;
its distance function $x\mapsto\dist(x;\Cantor)$
is computable by a strongly analytic machine without constants.
\end{exa}
\begin{proof}
Recall that $\Cantor$ consists of all real numbers $x\in[0,1]$
having a ternary expansion $x=\sum_{n=1}^\infty c_n 3^{-n}$
with $c_n\in\{0,2\}$ and is obviously closed.
\\
Concerning computability of $x\mapsto\dist(x;\Cantor)$,
observe that the sequence of open rational intervals
$I_{\langle m,k\rangle}:=\big((3k+1)/3^m,(3k+2)/3^m\big)$, 
$m\in\IN$, $k=0,\ldots,3^{m-1}-1$
is recursive and exhausts $[0,1]\setminus\Cantor$.
Therefore the sequence $y_m$,
defined by $y_m := \min\{|x-(3k+1)/3^m|,|x-(3k+2)/3^m|\}$
if $x\in I_{\langle m,k\rangle}$ for some $k$
and $y_m:=0$ otherwise, is computable from $x$
and has $\max\{y_1,y_2,\ldots,y_m\}$ converging
to $\dist(x;\Cantor)$ from below.
\\
In order to approximate 
$\dist(x;\Cantor)$ from above, 
observe that the sequence \linebreak
$x_{\bar c}:=(\sum_{n=1}^{|\bar c|} c_n3^{-n})_{_{\bar c\in\{0,2\}^*}}$
(w.r.t. lexicographically ordered index set) is 
recursive and dense in $\Cantor$.
Therefore, $z_m:=\min\{|x-x_{\bar c}|:\bar c\in\{0,2\}^{\leq m}\}$
is computable from $x$ and converges to $\dist(x;\Cantor)$ from above. 
\\
So $y_m\leq\sup_m y_m=\dist(x;\Cantor)=\inf_m z_m\leq z_m$
shows the existence of a subsequence $y_{m_n}$ with
$z_{m_n}-z_{m_n}\leq 2^{-n}$:
this can be sought for computationally
and satisfies $|y_{m_n}-\dist(y;\Cantor)|\leq2^{-n}$
as required.
\end{proof}

\subsection{Topological Complexity of \BSS and Strongly Analytic Computation}
While strongly analytic machines can compute strictly more
(e.g., the exponential function) than \BSS machines,
the present section reveals that the topological complexity
(in the sense of descriptive set theory) does not increase
for decision problems, and increases only slightly for 
function problems.

Recall that, classically, 
a (possibly partial) function $f:\subseteq\IN^*\to\IN^*$
is computable ~iff~ its graph is semidecidable.
The appropriate counterpart for a real function
$f:\subseteq\IR^*\to\IR$ are the \emph{strict epigraph} 
and \emph{strict hypograph} \cite{Emperor}:
\begin{gather*}
\sepigraph(f) \;:=\;
\big\{ (\vec x,y) : \vec x\in\dom(f), y> f(\vec x) \big\} , \\
\shypograph(f) \;:=\;
\big\{ (\vec x,y) : \vec x\in\dom(f), y< f(\vec x) \big\}  \enspace .
\end{gather*}
In the \BSS model, the square root has both 
strict epigraph and strict hypograph decidable
yet is not computable.

\begin{thm} \label{t:Graph}\hfill
\begin{enumerate}[\em a)] 
\item Conversely, if both $\sepigraph(f)$ and $\shypograph(f)$
are semidecidable by a \BSS machine with/out constants, then $f$ is
computable by a strongly analytic machine with/out constants.
\item
And if both $\sepigraph(f)$ and $\shypograph(f)$
are semidecidable \emph{in $\dom(f)$} by a \BSS machine with/out constants,
then \emph{some extension} of $f$ is computable by a strongly analytic machine
with/out constants. 
\item
A set $\IS\subseteq\IR^*$ is decidable by a \BSS machine
with/out constants ~iff~ its characteristic function
$\cf{\IS}:\IR^*\to\{0,1\}$ is computable by a strongly
analytic machine with/out constants.
\item
Every open subset of $\IR^k$ is \BSS-semidecidable.
\item
Every continuous (i.e. $\BorelS_1$-measurable) 
total function $f:\IR^d\to\IR^k$ is computable
by a strongly analytic machine;
\item and so is every countable family $f_\ell:\IR^d\to\IR^k$
($\ell\in\IN$) of continuous total functions.
\end{enumerate}
\end{thm}\vspace{3 pt}
\noindent Note the subtle mismatch between a) and b+c)
with respect to relative/absolute semidecidability,
imposed by Example~\ref{x:Domain}b). Similarly,
Item~f) does not extend to arbitrary partial functions.

We remark also that Item~f), together with Fact~\ref{f:Dist},
yields the distance function of the \aname{Mandelbrot Set} to
be computable by a strongly analytic machine---whereas its
computability in Recursive Analysis is still an open question
(and consequence of the \aname{Hyperbolicity Conjecture})
\cite{Hertling}; yet the \aname{Mandelbrot Set} cannot be
decided by a \BCSS \mycite{Theorem~2.4.2}{BCSS} nor 
(Item~d) by a strongly analytic machine.

\proof (Theorem~\ref{t:Graph})\hfill
\begin{enumerate}[a)]
\item 
Given $(\vec x,z)$ with $\vec x\in\dom(f)$, let the \BSS machine
simulate the strongly analytic computation of $y:=f(\vec x)$
and denote by $(y_n)$ the output sequence it generates,
i.e satisfying $|y_n-y|\leq2^{-n}$.
Now search for some $n$ with $y_n+2^{-n}<z$;
if found, $y<z$ follows and the machine accepts
$(\vec x,z)\in\sepigraph(f)$; and
conversely, in case $(\vec x,z)\in\sepigraph(f)$,
$y<z$ holds and there exists $n$ with $y_n+2^{-n}<z$.
\\
Note that for $\vec x\not\in\dom(f)$, such $n$
may or may not exist; hence the described procedure
semidecides $\sepigraph(f)$ only \emph{in $\dom(f)$}.
\item
Let $\calM$ and $\calN$ denote \BSS machines semideciding
$\sepigraph(f)$ and $\shypograph(f)$, respectively.
In order to approximate $f(\vec x)$ for given $\vec x$
up to error $2^{-n+1}$ simulate, for all $q\in\IQ$ in 
parallel, both $\calM$ on $(\vec x,q+2^{-n})$
and $\calN$ on $(\vec x,q-2^{-n})$.
If both terminate, $q-2^{-n}<f(\vec x)<q+2^{-n}$
justifies to output $q$; and, conversely, 
for $\vec x\in\dom(f)$, 
there is a rational approximation $q$ to $f(\vec x)$
up to error $2^{-n}$ for which both simulations terminate;
whereas for $\vec x\not\in\dom(f)$, both simulations stall
by hypothesis.
\item
Similarly to b), but now the behavior of $\calM$ and $\calN$
on $\vec x\not\in\dom(f)$ is undefined. Hence the search for
approximations may or may not yield a strongly converging
output sequence, i.e. produce $\bot$ or some value $y$ for 
$f(\vec x)$.
\item
If $\IS$ is decidable by a \BSS machine,
its characteristic function is \BSS-computable;
hence strongly analytic, recall Section~\ref{s:Analytic}.
\\
Conversely let $\cf{\IS}$ be computable by a strongly analytic machine $\calM$.
Upon input of $\vec x$, $\calM$ will thus output a sequence 
$(y_n)$ of reals with $|y_n-\cf{\IS}(\vec x)|\leq2^{-n}$.
Since $\cf{\IS}(\vec x)\in\{0,1\}$, this means $y_2$ 
uniquely exhibits whether $\vec x\in\IS$ or $\vec x\not\in\IS$
holds. A \BSS machine thus suffices to simulate $\calM$ for
the finitely many steps it takes to generate $y_2$ and then
output either $0$ or $1$ accordingly.
\item
The open rectangles $(\vec a,\vec b):=\{\vec x:\vec a<\vec x<\vec b\}$ 
with rational corners $\vec a,\vec b$ are well-known to 
form a base of the Euclidean topology on $\IR^k$;
that is, every open $V\subseteq\IR^k$ can be written as a countable union of 
certain such 
open rational rectangles $(\vec a_i,\vec b_i)$. Their corners' coordinates,
being integer fractions, can be encoded in binary into a single real number
\mycite{Lemma~2.3}{Cucker}. Perusing this as a pre-stored constant,
a \BSS machine can, given $\vec x\in\IR^k$, iteratively
extract the coordinates of the open rational rectangles
and search for one to contain $\vec x$.
\item
W.l.o.g. $k=1$.
By the Weierstra\ss{} Approximation Theorem, there exists a
double sequence $p_{n,m}$ of $d$-variate rational
polynomials such that $p_{n,m}$ converges to $f$ as $n\to\infty$
uniformly on $[-m,+m]^d$, w.l.o.g. with uniform error $\leq2^{-n}$.
Now encoding the list of rational coefficients into a real
constant as in d), a \BSS machine can, given $\vec x$ 
first find $m$ with $\vec x\in[-m,+m]^d$ and then
evaluate and output $p_{n,m}(\vec x)$ for $n=1,2,\ldots$.
\item
Similarly to f), encode the (still countable) 
list of rational coefficients of $p_{n,m,\ell}$ 
into a real constant.
\qed\end{enumerate}
%
Note that Item~e) is an extension of \mycite{Proposition~1}{Gaertner2}, where in the current case for the one-dimensional case all inputs in the open interval between two integers have the same computation path.
\begin{cor} \label{c:StrongBorel}\hfill
\begin{enumerate}[\em a)]
\item
Fix $X\subseteq\IR^d$. 
Every set $S\subseteq X$ 
\BSS-semidecidable in $X$
belongs to the Borel class $\BorelS_2(X)$.
\item
Every $S\subseteq X$ decidable in $X$ by a strongly analytic machine
belongs to $\BorelD_2(X)$.
\item
Each function $f:\subseteq\IR^d\to\IR$
computable by a strongly analytic machine 
is $\BorelS_2$-measurable in $\dom(f)$.
\item
For $f:\subseteq\IR^d\to\IR$
computable by a strongly analytic machine,
it holds $\dom(f)\in\BorelP_3(\IR^d)$.
\end{enumerate}
\end{cor}
Compare also \mycite{Section~4}{Cucker}.

\proof\hfill
\begin{enumerate}[a)]
\item
Note that every semialgebraic set is (a finite union
of basic semialgebraic sets and thus) the intersection
of a closed and an open set \cite[p.51 l.6]{BCSS}
and in particular in $\BorelS_2$, which is closed under countable unions. 
Now the \aname{Path Decomposition Theorem} for \BSS machines
\mycite{Theorem~2.3.1}{BCSS} shows that every \BSS-semidecidable
set (in $X$) is the countable union of semialgebraic sets (intersected with $X$).
\item
Decidability of $S$ means that both $S$ and its
complement are semidecidable, hence in $\BorelS_2$ by a);
that means $S\in\BorelS_2\cap\BorelP_2=\BorelD_2$.
\item
Let $V\subseteq\IR$ be open, $V=\bigcup_{n} (a_n,b_n)$ 
with $a_n,b_n\in\IQ$ enumerated by a \BSS machine
according to the proof of Theorem~\ref{t:Graph}e).
Now $f^{-1}[(a,b)]=f^{-1}\big[(a,\infty)\big]
\cap f^{-1}\big[(-\infty,b)\big]$
is semidecidable in $\dom(f)$ according to Theorem~\ref{t:Graph}a);
hence so is $f^{-1}[V]=\bigcup_n f^{-1}[(a_n,b_n)]$;
and thus $\BorelS_2$ in view of a).
\item
Let $\calM$ denote a strongly analytic machine computing $f:\subseteq\IR^*\to\IR$.
According to Equation~(\ref{e:Variation}), $\vec x\in\dom(f)$
~iff~ for all $n,m\in\IN$, $(\vec x,n,m)$ belongs to the set
\[ \big\{(\vec x,n,m):\calM \text{ on input $\vec x$ prints
$y_1,\ldots,y_n,\ldots,y_m$ with } |y_n-y_m|\leq2^{-n}+2^{-m} \big\} \]
which is clearly \BSS semidecidable and thus $\BorelS_2$ according
to a). Adding the universal quantification over $n,m$,
it follows that $\dom(f)$ is $\BorelP_3$.
\qed\end{enumerate}
The following result 
due to \person{Arno Pauly} (personal communication)
exhibits the topological difference between functions computable
by \BCSS machines and by strongly analytical ones,
recall Corollary~\ref{c:StrongBorel}c).

\begin{thm} \label{t:Arno}\hfill
\begin{enumerate}[\em a)]
\item
Every (possibly partial) function 
$f:\subseteq\IR^d\to\IR$ 
computable by a \BCSS machine is $\BorelD_2$-measurable
in $\dom(f)$.
\item
More generally,
let $\calF$ denote a family of continuous, partial real functions
$f:\subseteq\IR^{d_f}\to\IR$ of various arities $d_f\in\IN$ with
domains in $\BorelD_2$.
Let $\calR$ denote a family of 
$\BorelD_2$-measurable real relations $R\subseteq\IR^{k_R}$
of various arities $k_R\in\IN$.
Consider a uniform machine model over the structure
\footnote{We refrain from formally defining the intuitive 
but tedious concept of a (nonuniform) machine model over a
structure but refer, e.g., to
\mycite{\textsection 4.A}{Poizat} 
(which technically 
restricts to structures having only total functions);
cmp. also \mycite{\textsection 3}{Zucker}} 
$(\IR,\calF,\calR)$,
i.e. capable of performing a finite sequence of operations from 
$\calF$ and branchings based on tests from $\calR$.
Then any (possibly partial) function $g:\subseteq\IR^d\to\IR^k$
computable by such a machine is necessarily $\BorelD_2$-measurable
in $\dom(f)$.
\end{enumerate}
\end{thm}
\proof\hfill
\begin{enumerate}[a)]
\item
follows from b) with 
$\calF:=\{+,-,\times,\div\}$ and $\calR:=\{=,<\}$.
\item
Similarly to the proof of \mycite{Theorem~3.3.1}{BCSS},
the computation of such a machine can be unrolled into a
(possibly infinite) binary tree $\calT$: Each internal node
$u$ describes the branching based on the outcome of a test 
``$\vec y\in R_u$?'' ($R_u\in\calF$) of intermediate results $\vec y$;
intermediate results which arise from the input $\vec x$
evaluated on functions $g_u$ which are compositions 
of functions from $\calF$.
In particular, for a leaf $v$ of $\calT$, the set $G_v$
of inputs $\vec x\in\IR^d$ ending up in $v$ is the intersection 
$G_v=\bigcap\nolimits_u \big\{\vec x:g_u(\vec x)\in R_u\big\}$
with $u$ running over the (finitely many)
internal nodes from $\calT$'s root to $v$;
and the output printed in $v$ is of the form $g_v(\vec x)$.
This yields a disjoint decomposition
$g=\biguplus\nolimits_v g_v|_{G_v}$,
now with $v$ running over all (countably many) leaves of $\calT$.
In particular, $g^{-1}[Y]=\biguplus_v(g_v^{-1}[Y]\cap G_v)$
holds for any $Y\subseteq\IR^k$.
\\
Now by hypothesis, $g_v$ is continuous as the composition of
continuous functions; and for continuous $f_1,f_2$ with 
$\dom(f_1)$ and $\dom(f_2)$ both $\BorelD_2$-measurable,
$\dom(f_2\circ f_1)=\{\vec x:f_1(\vec x)\in\dom(f_2)\}$
is easily verified to be again $\BorelD_2$-measurable:
recall that $\BorelD_2$ is closed under both finite unions
and finite intersections. Similarly, it follows that each
$G_v\subseteq\IR^d$ is $\BorelD_2$-measurable as well.
Thus, both for open and for closed $Y\subseteq\IR^k$,
$g^{-1}[Y]$ is a countable union of $\BorelS_2$ sets.
\qed\end{enumerate}
Corollary~\ref{c:StrongBorel},
Theorem~\ref{t:Graph}e),
and Theorem~\ref{t:Arno}a)
are (almost) best possible:

\begin{exas} \label{x:BssBorel}\hfill
\begin{enumerate}[a)]
\item
The set $\IQ$ of rational numbers
is \BSS-semidecidable (but not in $\BorelP_2$).
\item
The characteristic function $\cf{[0,1)}:\IR\to\{0,1\}$
is \BSS-computable but is not $\BorelS_1\cup\BorelP_1$-measurable.
\item
\person{Cantor}'s \aname{Excluded Middle}
set $\Cantor\subseteq[0,1]$ belongs to $\BorelP_1\subseteq\BorelS_2$,
but is not \BCSS semidecidable.
\item
Recall \person{Thomae}'s or \aname{Popcorn Function} $h:\IR\to\IR$, 
defined as  $h(x)=0$ for $x\in\IR\setminus\IQ$
and $h(\pm p/q)=1/q$ for coprime $p,q\in\IN$, $h(0)=1$. \\
This function is computable by a strongly analytic machine
but is not $\BorelP_2$-measurable.
\item
There is a function $f:\subseteq\IR^2\to\{0\}$ computable by a strongly analytic machine
with $\dom(f)=\IQ\times(\IR\setminus\IQ)$ not $\BorelS_2\cup\BorelP_2$-measurable.
\end{enumerate}
\end{exas}
\proof\hfill
\begin{enumerate}[a)]
\item
A \BSS machine can, given $x\in\IR$,
enumerate all pairs $r,s\in\IZ$ and
compare $x=r/s$ to semidecide ``$x\in\IQ$''.
\item
The set $[0,1)$ is neither closed nor open, hence its characteristic
function is not $\BorelS_1\cup\BorelP_1$-measurable.
\item
Note that each singleton $\{x\}\subseteq\Cantor$ is a
connected component of $\Cantor$ of its own. Hence
$\Cantor$ has uncountably many connected components;
whereas any \BSS-semidecidable
set, being a countable union of semialgebraic sets
(recall the proof of Corollary~\ref{c:StrongBorel}a)
of only finitely many connected components each
\mycite{Section~5.2}{Basu}, can have at most countably many 
connected components.
\item
Recall from a) that $\IQ$ is not in $\BorelP_2$
but the preimage of an open set: $\IQ=h^{-1}[(0,2)]$.
We now describe a machine computing $h(x)$ on input $x>0$:
\\
Iteratively for $q=1,2,3,\ldots$ test whether $q\cdot x$ is an integer;
if not, output $q$ as approximation to $h(x)$ up to error $1/q$
and continue with the iteration; otherwise switch to outputting 
$1/q,1/q,1/q,\ldots$ as approximations to $h(x)$ up to error $1/m$ 
for all $m\geq q$.
\\
It is easy to convert this sequence $(y_q)_{_q}$ of approximations
up to error $1/q$ into a sequence $(z_n)_{_n}$ of approximations
up to error $2^{-n}$ by printing only the subsequence $(y_{2^n})_{_n}$.
\item
Consider a machine which, given $(x,y)$, first searches
(without output) for $p,q\in\IZ$ with $x=p/q$. When found, 
it starts similarly enumerating each $r_n\in\IQ$ and
printing $2^{-n}$ until (and if) arriving at one with $r_n=y$.
\qed\end{enumerate}
The rough conclusion of this subsection is that
both \BSS model and analytic machines lie slightly skew to the Borel Hierarchy,
having topological power strictly between $\BorelS_1$ and $\BorelS_2$;
and partial functions are even more skewed relative to the hierarchy.

\subsection{Composition of Strongly Analytic Machines} \label{s:Composition}
The analytic machine models presume the input to be given exactly
but produce only approximations to the output. It is thus reasonable
to expect that the composition of two functions computable by analytic
machines in general need not itself be computable by an analytic machine.
This has been proven for weakly analytic machines
in \mycite{Lemma~6}{Chadzelek}; cmp. \mycite{Corollary~2.4}{Gaertner}.
It is not surprising that we can establish the same
for strongly analytic machines in Proposition~\ref{p:Vassilis}b) below.
However the use of descriptive set theory is of interest because of
the new perspective it provides:
It is well-known that the composition 
of two $\BorelS_2$-measurable functions is in general 
no more than $\BorelS_3$-measurable \mycite{Corollary~3.9}{EffBorel};
whereas the composition of
a $\BorelS_2$-measurable function with a continuous one is
again $\BorelS_2$-measurable.
In view of Theorem~\ref{t:Graph}f) and Corollary~\ref{c:StrongBorel}c),
\person{Vassilios Gregoriades} (personal communication)
thus raised the natural question of whether the composition of
a strongly analytic and a continuous function is again strongly analytic.
A complete answer is given by the already mentioned

\begin{prop} \label{p:Vassilis}\hfill
\begin{enumerate}[\em a)]
\item
Let $g:\IR^k\to\IR^\ell$ be computable by a strongly analytic machine
and let $h:\IR^d\to\IR^k$ denote a continuous function.
Then $h\circ g$ is computable by a strongly analytic machine.
\item
There exists a function 
$g:\IR\to\IR$ and a continuous function $f:\IR\to\IR$,
both computable by strongly analytic machines without constants,
such that $g\circ f$ is not computable by a strongly analytic machine.
\end{enumerate}
\end{prop}
\noindent In particular, we obtain:
\begin{cor}
The class of total real functions 
computable by strongly analytic machines 
is not closed under composition.
\end{cor}
\proof (Prop.~\ref{p:Vassilis})\hfill
\begin{enumerate}[a)]
\item
We refine the proof of Theorem~\ref{t:Graph}f)
by storing, in addition to 
the coefficients of rational polynomials $p_{n,m}$ approximating
$h|_{[-m,+m]^d}$ up to error $2^{-n}$ 
also some moduli of uniform continuity, that is, 
integers $\mu_{n,m}$ subject to:
\begin{equation} \label{e:Moduli}
\vec y,\vec y'\in[-m,+m]^d, \quad |\vec y-\vec y'|\leq2^{-\mu_{n,m}}
\qquad\Longrightarrow\quad |h(\vec y)-h(\vec y')|\leq2^{-n} \enspace . 
\end{equation}
Now, given $\vec x$ and a desired precision $2^{-n}$, 
determine $m$ with $\vec x\in[-m,+m]^d$ and evaluate $\vec y:=g(\vec x)$:
By hypothesis, the strongly analytic machine computing $g$ can produce
an approximation $\vec y'$ up to precision $2^{-\mu_{n+1,m+1}}$.
Finally output $z:=p_{n+1,m+1}(\vec y')$ and verify
\[ |h\circ g(\vec x)-z| \quad\leq\quad
 |h(\vec y)-h(\vec y')|\;+\;|h(\vec y')-z|
\quad\overset{\text{(\ref{e:Moduli})}}{\leq}\quad 
 2^{-n-1}+2^{-n-1} \enspace . \]
\item
Let $g(0):=1$ and $g(y):=0$ for $y\neq0$ denote the characteristic function
of $\{0\}$. Let $f:=\dist(\cdot,\Cantor)$ denote the distance function of the
Cantor set, recall Example~\ref{x:DistCantor}.
Since $\Cantor$ is closed, it follows $g\circ f=\cf{\Cantor}$:
a function not computable by a strongly analytic machine
according to Example~\ref{x:BssBorel}c).
\qed\end{enumerate}
We now extend Theorem~\ref{t:Graph}f):

\begin{thm} \label{t:Composition}
Every $\BorelS_2$-measurable $f:\IR^d\to\IR$ 
can be expressed as the composition $f=h\circ g$ of
a strongly analytic $g:\IR^d\to\IR^\omega$ and a strongly
analytic partial function $h:\subseteq\IR^\omega\to\IR$.
\end{thm}
%

\proof
Being $\BorelS_2$-measurable means 
\begin{equation} 
f^{-1}\big[(q,p)\big] \;=\; \bigcup\nolimits_k A_{k,q,p},
\quad\text{ for }
q,p\in\IQ \; \text{ and closed } A_{k,q,p}\subseteq\IR^d \enspace .
\end{equation}
\begin{iteMize}{$\bullet$}
\item
Since $Q:=\{(q,p,k):k\in\IN,q,p\in\IQ,q<p\}$ is countable,
Theorem~\ref{t:Graph}g) yields
a strongly analytic machine computing the function
\[ \hat g:Q\times\IR^d \;\ni\; 
\big((q,p,k),\vec x\big)\;\mapsto \;
\dist(\vec x,A_{k,q,p})\;\in\;[0,\infty) \]
which we shall identify with the function 
$g:\IR^d\to\IR^Q$, 
$\vec x\mapsto\big((q,p,k)\mapsto\hat g(q,p,k,\vec x)\big)$.
Note that $\exists k:\hat g(q,p,k,\vec x)=0\Leftrightarrow q<f(\vec x)<p$.
\item
Now consider the function 
\begin{multline*}
h:\subseteq\IR^Q\to\IR, \quad
\delta\mapsto\sup\{q:\exists k,p:\delta(q,p,k)=0\}
\quad\text{with}\quad \dom(h)\;\;:=\;\;\\
\big\{\delta:Q\to[0,\infty)\;\big|\;
\sup\{q:\exists k,p:\delta(q,p,k)=0\}=\inf\{p:\exists k,q:\delta(q,p,k)=0\}\big\} 
\end{multline*}
and observe that 
$\delta:=g(\vec x)$ has
$\sup\{q:\exists k,p:\delta(q,p,k)=0\}=f(\vec x)=
\inf\{p:\exists k,q:\delta(q,p,k)=0\}$; hence
$\dom(h)\subseteq\range(g)$ holds
and, moreover, $(h\circ g)(\vec x)=f(\vec x)$.
\item
Finally, $h$ is computable by a strongly analytic machine:
Given $\delta\in\dom(h)$ and for each $n\in\IN$, search for
$q,p,k$ with $\delta(q,p,k)=0$
and $p-q\leq2^{-n}$ and, when found, print $q$, then continue
with $n+1$.
On the one hand such $(q,p,k)$ exist 
because, according to the hypothesis $\delta\in\dom(h)$, it holds
$\sup\{q:\exists k,p:\delta(q,p,k)=0\}=\inf\{p:\exists k,q:\delta(q,p,k)=0\}\big\}
=h(\delta)=:y$.
On the other hand such a tuple satisfies
$q<y<p\leq q+2^{-n}$, hence the output
sequence converges effectively to this $y$.
\qed\end{iteMize}
Note that we have silently extended the classical analytic machine
model to infinite dimensional arguments and values---which raises

\begin{question}
In Theorem~\ref{t:Composition},
can the infinite-dimensional intermediate results
be avoided? Can $h$ be chosen total?
How far up on the Borel hierarchy of measurability
do compositions of $k$ strongly analytic functions reach/cover?
\end{question}
Indeed, strongly analytic machines can encode infinite sequences 
into single reals and back; but a priori, each such operation incurs 
an additional machine, thus resulting in the composition
in Theorem~\ref{t:Composition} to become three-fold.

\section{Comparing Weakly and Strongly Analytic Machines} \label{s:Weakly}
It will turn out (Theorem~\ref{t:Main})
that weakly analytic machines are
essentially strongly ones equipped with oracle access
to the \BSS Halting problem.

We first record the following relativizations of
Theorem~\ref{t:Graph} and Corollary~\ref{c:StrongBorel}:

\begin{cor} \label{c:RelativeBSS}\hfill
\begin{enumerate}[\em a)]
\item
Each set $\IS\subseteq\IR^d$ \BSS-semi\-decidable
with oracle $\IH$ necessarily belongs to Borel class $\BorelS_3$.
\item
Every $\IS\in\BorelS_2$ is \BSS-semidecidable with oracle $\IH$.
\item
Each total function $f:\IR^d\to\IR$
computable by a strongly analytic machine with $\IH$-oracle
is $\BorelS_3$-measurable.
\item
Every $\BorelS_2$-measurable total function $f:\IR^d\to\IR$
is computable by a strongly analytic machine with $\IH$-oracle.
\end{enumerate}
\end{cor}
\proof\hfill
\begin{enumerate}[a)] 
\item
follows from Corollary~\ref{c:StrongBorel}a) and 
Fact~\ref{f:Cucker}b), observing that 
$\{(\vec x,y):\forall z\in\IN:(\vec x,y,z)\in\IW\}$
is in $\BorelP_2$ because its complement is \BSS-semidecidable.
\item
Let $\IS=\bigcup_n \IA_n$ with $\IA_n$ closed,
i.e. the complement of $\IA_n$ is of the form
$\bigcup_m (\vec a_{n,m},\vec b_{n,m})$ with rational
corners $\vec a_{n,m},\vec b_{n,m}$. Just like
in the proof of Theorem~\ref{t:Graph}e),
this rational double sequence can be encoded into 
one single real constant in order to enable a \BSS
machine deciding 
$\IW:=\{(\vec x,n,m):\vec x\not\in(\vec a_{n,m},\vec b_{n,m})\}$.
Now apply Fact~\ref{f:Cucker}c) to
$\bigcup\nolimits_n \IA_n=\big\{ \vec x\big| 
\exists n \forall m : (\vec x,n,m)\in\IW \big\}$.
\item
Like in the proof of Corollary~\ref{c:StrongBorel}c)
and relativizing Theorem~\ref{t:Graph}a),
we observe that $f^{-1}[\bigcup_n(a_n,b_n)]$ 
is semidecidable by a \BSS machine with $\IH$-oracle. Now apply a).
\item
For $q\in\IQ$ and $n\in\IN$, consider the
$\BorelS_2$-set $f^{-1}[(q-2^{-n},q+2^{-n})]$.
Extending Item~b), we see that these sets are
\BSS-semidecidable with $\IH$-oracle 
\emph{uniformly} in $q$ and $n$.
Hence given $\vec x$, an $\IH$-oracle machine
can search and output, for each $n\in\IN$, some
$q\in\IQ$ with $\vec x\in f^{-1}[(q-2^{-n-1},q+2^{-n-1})]$.
\qed\end{enumerate}
\mycite{Theorems~2.15+2.16}{Cucker} establishes two natural problems over the reals as
\BSS-equivalent for (i.e. many-one reducible from and to) $\IH^\IH$. 
The next section will add \emph{Boundedness}; and Theorem~\ref{t:Convergence} shows
(the complement of) \emph{Convergence} \BSS-equivalent to $\IH^{\IH^\IH}$.

\subsection{The Boundedness Problem and Weak
  Semidecidability} \label{s:WeakSemi}\hfill

\noindent Consider the boundedness problem
for analytic machines:
\[ \IB \;\;:=\;\;
\big\{ (\langle\calM\rangle,\vec x):
\text{machine $\calM$ produces on input $\vec x$
some bounded sequence $(\vec y_n)_{_n}$} \big\} \enspace . \]
By convention, we regard also a finite sequence as bounded.

\begin{prop} \label{p:Boundedness}\hfill
\begin{enumerate}[\em a)]
\item
A \BSS machine with oracle access to $\IH$
can semidecide $\IB$
\item
but cannot decide $\IB$. More precisely, it holds
$\IH^\IH\BssReduceq\IB$.
\end{enumerate}
\end{prop}
\proof\hfill
\begin{enumerate}[a)]
\item
Given $\calM$ and $\vec x$,
iteratively try the bounds $n=1,2,\ldots$
and use oracle access to $\IH$ in order to 
detect whether some output of $\calM$ on $\vec x$
has norm exceeding $n$:
If so, retry with $n+1$; otherwise accept.
\item
Since $\IH^\IH$ is semidecidable relative to $\IH$,
it has the form of Equation~(\ref{e:Sigma2}).
Now consider the \BSS machine $\calM$ executing the
following algorithm:
Given $\vec x$ and iteratively for each 
$y=1,2,\ldots$, $\calM$ looks for some $z=1,2,\ldots$
such that $(\vec x,y,z)\not\in\IW$.
If such a $z$ is found, 
$\calM$ outputs $y$ and restarts with $y+1$;
otherwise $\calM$ keeps looking for $z$ indefinitely.
\\
For $\vec x\in\IH^\IH$, the above machine will thus 
eventually find an $y$ that leads to
an infinite loop on $z$; and hence a bounded
(even finite) output sequence.
Whereas for $\vec x\not\in\IH^\IH$, every $y\in\IN$ will
eventually be output by $\calM$, that is, an
unbounded sequence.
\qed\end{enumerate}
Theorem~\ref{t:Graph}d) suggests the following

\begin{defi} \label{d:Weak}
A set $\IS\subseteq\IR^*$ is
\aname{weakly decidable} ~iff~
its characteristic function $\cf{\IS}:\IR^*\to\{0,1\}$
is computable by a weakly analytic machine.
\\
$\IS$ is \aname{weakly semidecidable}
~iff~ $\IS$ is \BSS-semidecidable with $\IH$-oracle.
\end{defi}
In view of Fact~\ref{f:Cucker} and Proposition~\ref{p:Boundedness}, 
$\IS$ is weakly semidecidable ~iff~ it is \BSS many-one reducible
to $\IH^\IH$ or, equivalently, to $\IB$.

\begin{exas} \label{x:Weak}\hfill
\begin{enumerate}[a)]
\item
 For a function $f:\IR^*\to\IR^k$ 
 computable by a weakly analytic machine and
 for open $V\subseteq\IR^k$, the pre-image $f^{-1}[V]\subseteq\IR^*$
 is weakly semidecidable.
\item
 Every $\BorelS_2$-set is weakly semidecidable.
\end{enumerate}
\end{exas}
\proof 
Let $(\vec y_n)$
be a sequence output by the weakly analytic machine
on input $\vec x$, i.e. with $\lim_n\vec y_n=\vec y:=f(\vec x)$.
In view of Theorem~\ref{t:Graph}e), we can assume
to have an enumeration $(V_m)$ of rational open rectangles
exhausting $V=\bigcup_m V_m$ at our disposition.
\begin{enumerate}[a)]
\item
For each $m,k=1,2,\ldots$ test whether it holds
that the rectangle $[\vec y_n-2^{-m},\vec y_n+2^{-m}]$
is contained in $V_k$ for all
$n\geq m$. This can be achieved by setting up a
machine $\calN$ searching for a counter-example $n$
and querying oracle $\IH$ for non-termination of $\calN$.
If so, since $\vec y\in[\vec y_n-2^{-m},\vec y_n+2^{-m}]$
for all sufficiently large $n$, it follows $\vec y\in V$
and we can safely accept.
Conversely in case $\vec y\in V_k$, 
it holds $[\vec y_n-2^{-m},\vec y_n+2^{-m}]\subseteq V_k$
for all sufficiently large $n,m$; hence the above search
succeeds.
\item
Let $V=\bigcup_{j} A_j\in \BorelS_2$. Analogously to the proof of Theorem~\ref{t:Graph}e), the closed sets $A_j$ can be represented as complements of a countable union of open rectangles with rational corners $A_j=\left(\bigcup_{i} (\vec a_{j,i},\vec b_{j,i})\right)^\complement$. The rational coordinates of $\vec a_{j,i}$ and $\vec b_{j,i}$ can all be encoded into one real constant. A machine that semidecides $\vec x\in V$ tries, for increasing $n=1,2,\ldots$, whether $\vec x\in A_n$. To this end, the coordinates of the rectangles exhausting $A_n$ are extracted and for increasing $m$ the condition $\vec x\in(\vec a_{n,m},\vec b_{n,m})$ is checked. After each check, the machine outputs $n$, and as soon as a rectangle containing $x$ is found, the machine proceeds to the next $n$. If $\vec x \in V$, then it is in some $A_n$, and therefore in no rectangle $(\vec a_{n,m},\vec b_{n,m})$, $m\in\IN$. In this case, the machine never exceeds stage $n$. If, on the other hand, $\vec x \not\in V$, then for each $n$ there is such a rectangle, and the machine reaches (and outputs) each $n\in\IN$.
\qed\end{enumerate}

\subsection{Weakly Analytic Machines are the Jump of Strongly Analytic
  Ones}\hfill

\noindent Definition~\ref{d:Weak} is justified by Item~a) of the following

\begin{thm} \label{t:Main}\hfill
\begin{enumerate}[\em a)]
\item
$\IS\subseteq\IR^*$ is weakly
decidable 
~iff~ both $\IS$ and its complement are
weakly semidecidable.
\item
If a (possibly partial) function $f:\IR^*\to\IR^*$
is computable by a weakly analytic machine,
$f$ is also computable by a strongly analytic
machine with oracle access to $\IH$.
\item
Conversely, if $f$ is computable by a strongly analytic
machine with oracle access to $\IH$,
then some extension of $f$ is computable by a weakly analytic machine.
\end{enumerate}
\end{thm}
\noindent The equivalence in b+c) constitutes an analytic analogue
of the \aname{Shoenfield Limit Lemma}. The slight mismatch
with respect to partial functions resembles
Theorem~\ref{t:Graph}abc) and raises

\begin{question}
Is every partial function computable by a strongly analytic machine
with oracle access to $\IH$ also
computable by a weakly analytic machine?
\end{question}

\proof (Theorem~\ref{t:Main})\hfill
\begin{enumerate}[a)]
\item
Suppose $\calM$ is a weakly analytic machine computing $\cf{\IS}$,
and let $(y_n)$ denote the sequence output by $\calM$ on input $\vec x$.
We show that both $\IS$ and its complement $\IS^\complement$ are
reducible to $\IB$. To this end modify $\calM$ to output 
$u_n:=1/\max\{y_n,1/n\}$: Since $\{0,1\}\ni\lim_n y_n$ exists,
$u_n$ is bounded ~iff~ $y_n\to1$ ~iff~ $\vec x\in\IS$.
Similarly, $v_n:=1/\max\{1-y_n,1/n\}$ is bounded ~iff~ $y_n\to0$ ~iff~ $\vec x\not\in\IS$.
\\
Conversely consider \BSS machines $\calM$ and $\calN$
computing reductions from $\IS$ and $\IS^\complement$ to $\IB$, respectively.
Then the following machine weakly computes $\cf{\IS}$:
Given input $\vec x$, test (by parallel simulation) for increasing bounds $n=1,2,\ldots$ 
whether some output of $\calM$ or of $\calN$ exceed the bound $n$. 
If so, append ``$0$'' to the output if it was $\calM$,
and ``$1$'' if it was $\calN$; then increment the bound $n$.
Since $\vec x$ belongs to exactly one of $\IS$ and $\IS^\complement$,
precisely one of $\calM,\calN$ produces an unbounded sequence;
and our output thus becomes a stationary sequence of $1$s 
(in case $\vec x\in\IS$) or of $0$s ($\vec x\in\IS^\complement$), respectively.
\item
Let $\calM$ denote a weakly analytic machine computing $f$
and $(\vec y_n)$ the (possibly finite) sequence output on input $\vec x$.
We describe another machine $\calN$ that uses oracle queries to $\IH$ 
in order to output a subsequence of $(\vec y_n)$ satisfying Equation~(\ref{e:Variation}).
Note that the violation of this condition can be detected by searching for $n,m$
and hence is semidecidable. Using $\IH$, $\calN$ can thus decide, 
for each required precision index $k=1,2,\ldots$ and each $K\in\IN$,
whether $\|\vec y_K-\vec y_m\|\leq2^{-k}+2^{-m}$ holds for all $m$.
On the other hand such $K=K(k)$ exists to every $k$
~iff~ $(\vec y_n)$ converges.
$\calN$ will thus, iteratively for $k=1,2,\ldots$, search for such
a $K$ and, when found, output the corresponding $\vec y_K$.
Note that, if $\calM$ outputs only a finite sequence,
so will $\calN$.
In effect, the subsequence printed by $\calN$ satisfies
the bottom of Equation~(\ref{e:Analytic}) ~iff~ the 
original sequence printed by $\calM$ satisfies the top of
Equation~(\ref{e:Analytic}).
\item
Assume $f$ is computed by the strongly analytic machine $\calM$ with oracle access to $\IH$. We describe a weakly analytic machine $\calN$ that computes $f$. Fix an input $\vec x\in\IR^*$. For $\sigma\in\{0,1\}^\IN$, let $\vec y_n^\sigma$ be the output of the machine $\calM$, simulated under the assumption that the $j$-th oracle query is negative or positive, depending on $\sigma(j)$. For increasing $n=1,2,\ldots$ (\emph{simulation level}), the machine $\calN$ simulates (without output) $\calM$ up to the $n$-th output. It simulates all machines queried by the oracle, up to $n$ output steps (or until they halt), and stores the knowledge about the oracle answer in the sequence $\sigma_n\in\{0,1\}^\IN$ (0: does not halt, 1: halts), initially assuming all oracle queries to be answered negatively. In addition, the conditions
\begin{equation} \label{e:VariationSigma}
\|\vec y_i^{\sigma_n}-\vec y_j^{\sigma_n}\| \;\leq\; 2^{-i}+2^{-j} \quad\forall 1\leq i<j\leq n
\enspace .
\end{equation}
are checked. If one of these conditions is violated, the number of steps of all simulated oracle queries is increased until all these conditions are fulfilled. As soon as this is the case, $\calN$ outputs $\vec y_n^{\sigma_m}$, $m$ being the number of simulated steps of the oracle queries. Then, $\calN$ proceeds to level $n+1$.\\
Given $\vec x\in \dom(f)$ and $N\in\IN$, there is a number of simulation steps $n_0(N)$ after which all oracle queries made until output $N$ of $\calM$ have been answered correctly. At level $n\geq n_0(N)$, $\calM$ produces an output $\vec y_n^{\sigma_n}$ which, because of Equation~(\ref{e:VariationSigma}), satisfies $\|\vec y_N^{\sigma_n}-\vec y_n^{\sigma_n}\|\leq2^{-N}+2^{-n}$. Furthermore, because at level $n$, all oracle assumptions up to output $N$ are correct, we know that $\vec y_N^{\sigma_n} = \vec y_N$. Therefore, the outputs of $\calN$ correctly converge to $\vec y$.
\qed
\end{enumerate}
%
In connection with Corollary~\ref{c:RelativeBSS}, we conclude

\begin{cor} \label{c:Weak}\hfill
\begin{enumerate}[\em a)]
\item
Every weakly semidecidable $\IS\subseteq\IR^d$ 
belongs to Borel class $\BorelS_3$.
\item
Every function $f:\IR^d\to\IR^k$
computable by a weakly analytic machine
is $\BorelS_3$-measurable.
\item
Conversely, every $\BorelS_2$-measurable $f:\IR^d\to\IR^k$
is computable by a weakly analytic machine.
\end{enumerate}
\end{cor}
\noindent Again, Corollary~\ref{c:Weak} is (almost) best possible:

\begin{exa} \label{x:BssBorel2}
The set $\IQ\times(\IR\setminus\IQ)$ 
is decidable by a weakly analytic machine
(but not in $\BorelP_2\cup\BorelS_2$).
\end{exa}
\begin{proof}
Since $\IQ$ is semidecidable by a \BSS machine,
it is decidable relative to $\IH$;
and so is $\IR\setminus\IQ$.
$\IQ\times(\IR\setminus\IQ)$ can be decided
relative to $\IH$ by testing both components
separately; hence this set is weakly decidable 
according to Theorem~\ref{t:Main}a).
\end{proof}

\begin{question}
Is there a set $S\subseteq\IR$
weakly semidecidable yet
such that $S\not\in\BorelP_3$?
\end{question}

\subsection{The Convergence Problem and Na\"ive Semidecidability} \label{s:Convergence}
Our proof of Example~\ref{x:Weak}a) erroneously accepts in case the
output sequence $\vec y_n$ fails to converge by having several 
accumulation points all contained in some $V_m$. This cannot happen
for $\vec y_n$ produced by the weak evaluation of a \emph{total}
function $f$. 

\begin{defi} \label{d:Convergence}
In view of the second part of Equation~(\ref{e:Analytic}), consider
\[ \IK :=
\big\{ (\langle\calM\rangle,\vec x):
\text{machine $\calM$ produces on input $\vec x$
some convergent infinite sequence $(\vec y_n)_{_n}$} \big\} \]
Call a set $\IS\subseteq\IR^*$ \aname{na\"ively semidecidable}
if there is a weakly analytic machine calculating
(i.e. printing a sequence of approximations which converge to)
\begin{enumerate}[i)]
\item the real number 0 for inputs $\vec x\in\IS$
\item $\bot$ (i.e. fails to converge) for inputs $\vec x\not\in\IS$.
\end{enumerate}
\noindent
A machine that produces only finitely many outputs is considered divergent.
\end{defi}
Diagonalization shows \cite{Vierke} 
that $\IK$ is undecidable
to a weakly analytic machine;
yet it can be written as the composition
of two functions computable by
weakly analytic machines \mycite{Theorem~2.3}{Gaertner}.

\begin{thm} \label{t:Convergence}\hfill
\begin{enumerate}[\em a)]
\item
If $\IS\subseteq\IR^*$ is 
na\"ively semidecidable, it is
\BSS many-one reducible to the convergence problem $\IK$.
\item
Conversely, every $\IS\subseteq\IR^*$ 
\BSS many-one reducible to $\IK$ 
is na\"ively semidecidable.
\item
The complement of $\IK$ is \BSS many-one 
reducible to $\IH^{\IH^\IH}$.
\item
Conversely, $\IH^{\IH^\IH}$ is
\BSS many-one reducible to $\IK^\complement$.
\end{enumerate}
\end{thm}
Since $\IH^{\IH^\IH}$ is strictly harder than $\IH^\IH$,
convergence is strictly harder than
boundedness; and weak semidecidability 
is strictly stronger a notion than 
na\"ive semidecidability.

\proof\hfill
\begin{enumerate}[a)]
\item
Let $\calM$ na\"ively semidecide $\IS$.
Then $\vec x\mapsto(\langle\calM\rangle,\vec x)$ constitutes a
\BSS-computable many-one reduction of $\IS$ to $\IK$:
For $\vec x\in\IS$, $\calM$ on input $\vec x$ outputs a
sequence converging to 0; whereas for $\vec x\not\in\IS$,
$\calM$ on input $\vec x$ outputs a divergent sequence.
\item
Consider a many-one reduction from
$\IS$ to $\IK$, i.e. mapping an instance $\vec u$ for $\IS$ 
to an instance $(\langle\calM\rangle,\vec x)$ of $\IK$. 
Consider a \BSS machine which simulates $\calM$ and
replaces its output sequence $(y_n)$ by
the sequence $(|y_n-y_m|)_{_{\langle n,m\rangle}}$,
where $\langle\,\cdot\,,\,\cdot\,\rangle:\IN\times\IN\to\IN$
denotes a recursive pairing function.
To see that this machine na\"ively semidecides $\IS$,
observe that $(y_n)$ converges (i.e. is Cauchy) ~iff~
$(|y_n-y_m|)_{_{\langle n,m\rangle}}$ converges to 0: 
To every $N\in\IN$ there is some $M\in\IN$ such that
$n,m\geq N$ implies $\langle n,m\rangle\geq M$; and, conversely,
to every $M\in\IN$ there is some $N\in\IN$ such that
$\langle n,m\rangle\geq M$ implies $n,m\geq N$. 
\item
We employ from \mycite{Theorem~2.11}{Cucker}
the following extension of Fact~\ref{f:Cucker}b+c):
\begin{quote}
$\IS\subseteq\IR^*$ is \BSS many-one reducible to
$\IH^{\IH^\IH}$ ~iff~ there exists some \BSS-decidable 
$\IW\subseteq\IR^*$ such that
\begin{equation} \label{e:Sigma3}
\IS\;=\;\big\{ \vec x\in\IR^* \;\big|\; \exists u\in\IN
\;\forall v\in\IN\;\exists w\in\IN: \; (\vec x,u,v,w)\in\IW\big\} 
\enspace . \end{equation}\end{quote}
Now observe that an infinite real sequence $(y_n)$ fails to converge
~iff~
\[ \exists k\in\IN \;\forall K\in\IN\;\exists \langle n,m,\ell\rangle\in\IN:
\quad n,m\geq K \;\wedge\; |y_n-y_m|\geq1/k \enspace . \]
Finally, to take into account the computation of $(y_n)$,
consider the \BSS-decidable
\begin{multline*} 
\IW \;:=\;\big\{ (\langle\calM\rangle,x,k,K,n,m,\ell):
\vec x\in\IR^*,k,K,n,m,\ell\in\IN;\;n\geq m\geq K \text{ and} \\
\text{$\calM$ on input $\vec x$ within $\ell$ steps
outputs $y_1,\ldots,y_m,\ldots,y_n$ with } |y_n-y_m|\geq1/k \big\} 
\end{multline*}
\item
Again we invoke the characterization from \mycite{Theorem~2.11}{Cucker}
and show that every $\IS$ of the form (\ref{e:Sigma3}) is 
\BSS many-one reducible to $\IK^\complement$.
To this end execute the following procedure for each
$u\in\IN$ in parallel: 
\begin{quote}
Let $v:=1$ and, for each $w=1,2,\ldots$
output ``$0$''. Moreover, if $(\vec x,u,v,w)\in\IW$,
output ``$2^{-u}$'', increment $v$, and restart with $w=1,2,\ldots$
\end{quote}
Observe that, if $\forall v\exists w: (\vec x,u,v,w)\in\IW$ holds,
this will for each such $u$ produce a sequence with accumulation points precisely $0$ and $2^{-u}$;
and otherwise a sequence containing finitely many $2^{-u}$'s and $0$'s otherwise.
Hence, if $\exists u\forall v\exists w: (\vec x,u,v,w)\in\IW$ holds,
the parallel search for such $u$ will result in a non-converging output;
and otherwise in an output converging to 0.
\qed\end{enumerate}

\section{Conclusion}
In Recursive Analysis, adding oracle access\footnote{%
in the sense of querying digits of a single infinite sequence.
As a referee kindly pointed out, this corresponds more to a
pre-stored constant of a \BSS machine than to real number
oracle queries. Other notions of oracles in Recursive Analysis
are discussed in \cite{ClosedChoice}.}
(to the, say, Halting Problem)
does not increase the topological power of computation:
Computable real functions are still necessarily continuous
(i.e. $\BorelS_1$-measurable).
Relaxing the output representation from approximations
with error bounds to converging approximations without error bounds,
however, does increase the topological capabilities by
proceeding one step up the (effective) Borel Hierarchy.

For Analytic Machines, on the other hand, we have revealed
both to be equivalent: relaxing output with to without 
error bounds ~and~ permitting oracle access to the
\BSS Halting Problem. Both amount to climbing up one step 
in the (non-effective) Borel Hierarchy---although the 
algebraic model lies slightly skewly to its levels.

\begin{question}
How about degrees of quasi-strongly analytic machines?
\end{question}
These are a blend of weak and strong ones,
required to provide
error bounds which, however, they are permitted to violate a
finite (yet unbounded) number of times.


\end{document}